\documentclass[journal,comsoc]{IEEEtran}

\usepackage{graphicx}
\usepackage{subcaption}
\usepackage{epstopdf}
\usepackage{amsmath}

\usepackage{algorithmicx}
\usepackage{algpseudocode}
\usepackage{algorithm}
\usepackage{float}

\usepackage{calc}
\usepackage{graphics,amsmath,amsfonts,amssymb,setspace,exscale,relsize,makeidx}
\usepackage{tikz}
\usepackage{amsthm}

\newtheorem{theorem}{Theorem}

\theoremstyle{remark}

\begin{document}

\title{The More the Merrier: Enhancing Reliability of 5G Communication Services with Guaranteed Delay}

\author{Prabhu Kaliyammal Thiruvasagam, Vijeth J Kotagi, and C Siva Ram Murthy \\ 
	Indian Institute of Technology Madras, Chennai 600036, India \\
	prabhut@cse.iitm.ac.in, vijethjk@gmail.com, murthy@iitm.ac.in}

\maketitle

\begin{abstract}
Although network functions virtualization and software-defined networking offer many dynamic features such as flexibility, scalability, and programmability for easy provisioning of services at a lesser cost and time through service function chaining, it introduces new challenges in terms of reliability, availability, and latency of services. Particularly, softwarization of network and service functions (e.g., virtualization, anything as a service, dynamic virtual chaining, and routing) impose high possibility of network failures due to software issues than hardware. In this letter, we propose a novel solution called eRESERV\let\thefootnote\relax\footnote{This work was supported by the Department of Science and Technology, New Delhi, India.} to enhance the reliability of service chains in 5G while meeting the service level agreements.
\end{abstract}


\IEEEpeerreviewmaketitle

\section{Introduction}

Softwarization in 5G networks to support services such as enhanced mobile broadband and ultra-reliable low-latency communications has revolutionized the networking industry. It is expected that 5G networks will meet the stringent requirements of communication services and business models of 2020 and beyond \cite{5g_req}. Network Functions Virtualization (NFV) and Software-Defined Networking (SDN) are considered as the key technology enablers for softwarization in 5G networks \cite{zarrar}.

NFV allows network functions (or middleboxes) to run as software modules on commercial-off-the-shelf servers rather than running on specialized hardware appliances. These virtualized software modules are called as Virtual Network Functions (VNFs). NFV leverages virtualization, cloud computing, and SDN technologies to provide anything as a service (e.g., core network as a service, security as a service, etc.) dynamically over the network, and reduces capital expenditures and operational expenditures.

Traditionally, network/communication services are provided through one or more network functions to deliver an end-to-end (e2e) service. Service Function Chaining (SFC) or simply service chain involves instantiation of an ordered list of network/service functions (e.g., firewall, IDS, and proxy), and connecting them together as a chain of network functions to provide the e2e services \cite{Haeffner-SFC}. NFV facilitates easy provisioning of services by dynamically placing VNFs in the virtual environment and chaining them together as an SFC. As a 5G network aims at providing services to diverse industry verticals, tens to hundreds of VNFs are placed on a set of servers and chained to create multiple SFCs. 

Although NFV and SDN provide many benefits in terms of cost reduction and flexible management of resources to 
dynamically provide diverse services, they create avenues for reliability, availability, and latency related issues as they are prone to software failures. Particularly, softwarization of network and service functions impose high possibility of network failures due to software issues than hardware. For instance, failure of a VNF or a virtual link in an SFC will bring down the entire chain and disrupt the service which may result in customer dissatisfaction and revenue loss. Failures may happen both at the substrate network and virtual network, but the frequency of failures at virtual network is higher than substrate networks \cite{Benz}. Generally, failures in virtual network may happen due to software bugs, API failures, incorrect specifications, network design flaws, improper testings, and network operator errors.

Another important aspect of NFV is meeting Service Level Agreements (SLAs) in terms of delay and resources of all services. A common approach to achieve higher reliability and meeting delay constraints is by placing redundant network elements (also called as backup) \cite{Rahman}\cite{Qu-1}\cite{Sun}. However, placing redundant network elements are expensive and ineffective in terms of effective utilization of available resources. 

Therefore, in this letter, we propose a novel solution dubbed \textit{e}nhancing \textit{RE}liability of \textit{SERV}ice chain (eRESERV), which enhances the reliability of an SFC while meeting the delay constraints of the SFC. The proposed solution also minimizes the resource requirement of an SFC to achieve higher reliability without compromising the delay constraints. By extensive simulations we show the effectiveness of the proposed solution in terms of reliability, expected response time, and resource requirement when compared to traditional backup settings. We further analyze our solution using queuing theory by modeling an SFC as M/M/1 and M/M/m tandem network of queues. 

\section{Network Model}
We represent the physical/substrate network as a graph $G_p=(\mathcal{N},\mathcal{L})$, where $\mathcal{N}$ represents the set of physical nodes and $\mathcal{L}$ represents the set of physical links. Physical resources are virtualized to create virtual networks and controlled with the help of hypervisor (or virtual machine manager) and SDN controller. SFCs are created at the virtual network in the cloud data center to provide service for various network service requests. We represent the virtual network as a graph $G_v = (\mathcal{V}, \mathcal{E})$, where $\mathcal{V}$ represents a set of VNFs (e.g., load balancer, firewall, intrusion detection system, proxy, mobility management entity, serving/packet gateway, home subscriber server, etc.) and $\mathcal{E}$ represents a set of virtual links in the system.
VNFs are hosted on the physical servers and virtual links are created to interconnect VNFs and carry virtual network traffic over the physical links. The physical and virtual network resources together form NFV infrastructure (NFVI), which is managed and controlled by Virtual Infrastructure Manager (VIM) along with SDN controller. 

We assume that service providers offer a finite set of services using SFCs. Let the set of all SFCs provided by a service provider be denoted by $\mathcal{S}$. Each SFC $s \in \mathcal{S}$ provides a particular service and is represented as an acyclic directed graph $G_s = (\mathcal{V}_s, \mathcal{E}_s)$, where $\mathcal{V}_s$ and $\mathcal{E}_s$ represent the set of VNFs in sequential order and the set of links that interconnect these VNFs, respectively. For example, consider a web service request $s$, where the set of VNFs $\mathcal{V}_s$ required to cater the service $s$ in an order are firewall, intrusion detection system, and proxy. As each SFC provides a particular service, we use terms SFC and service interchangeably.

In this letter, we consider that any service request $s \in \mathcal{S}$ has a latency requirement denoted by $\Psi_s$, and is considered to have an arrival rate $\lambda_s$ which follows Poisson distribution. Each VNF $v \in \mathcal{V}$ is considered to have a processing rate of $\mu_v$ and follows an exponential distribution with corresponding response (both waiting and processing) time $\psi_v$. 

Now our aim is to design an SFC to provide service for a service request such that the SFC offers high reliability, meets the delay constraint $\Psi_s$, and request is satisfied with the minimal resources. In this letter, we consider number of virtual cores (directly relates to $\mu_v$), assigned to all VNFs in $\mathcal{V}_s$ to process the traffic of requested service, as the resources.

\section{Enhancing Reliability of an SFC}
\label{sec:enhanceRel}
Consider an SFC $s$ as shown in Fig. \ref{Fig:sc}, which requires four VNFs, i.e., $|\mathcal{V}_s| = 4$ and has arrival rate $\lambda_s$. Each VNF $v$ is reliable with probability $p_v$. If an SFC $s$ provides a service for a service request which has delay constraint of $\Psi_s$, then $\sum\limits_{v \in \mathcal{V}_s} \psi_v \le \Psi_s$. The resource requirement of $s$ is 
$\gamma_s = \sum\limits_{v \in \mathcal{V}_s} \gamma_v$. If each VNF $v \in \mathcal{V}_s$ is reliable with a probability $p_v$, then the reliability $r_s$
can be calulated as,
\begin{align}
r_s &= \prod_{v \in \mathcal{V}_s} p_v \label{eq:relSFC}
\end{align}

Now consider a backup setting where each VNF is provided with a dedicated $b$ number of backups. Fig. \ref{Fig:scb}
 shows one example where each VNF is provided with one backup. Now, reliability $p_v(b)$ of each VNF $v \in \mathcal{V}_s$ can be calculated as,
\begin{equation}
p_v(b) = 1 - (1-p_v)^b \times (1-p_v) = 1 - (1-p_v)^{b+1}
\label{eq:relOfVNFWithBackup}
\end{equation}
Now, by substituting above equation in Equation \ref{eq:relSFC}, the new reliability of the SFC is,

\begin{equation}
r_s(b) = \prod_{v \in \mathcal{V}_s} p_v(b) = \prod_{v \in \mathcal{V}_s} (1 - (1-p_v)^{b+1})
\label{equ:scb_rel}
\end{equation}
Clearly, $r_s(b) > r_s~\forall b \ge 1$. However, the number of resources now required for the SFC $s$ will be 
\begin{equation}
\gamma_s = (b+1) \times \sum\limits_{v \in \mathcal{V}_s} \gamma_v
\label{equ:scb_rc}
\end{equation}

\begin{figure}[t]
	\centering
	\includegraphics [width=\columnwidth, height=1.3cm]{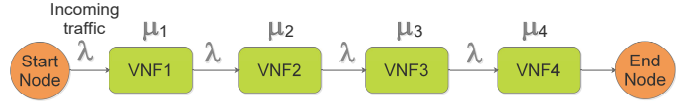}
	\caption{An ordered sequence of VNFs in an SFC.} 				
	\label{Fig:sc}
\end{figure}
\begin{figure}[t]
	\centering
	\includegraphics [width=\columnwidth, height=2cm]{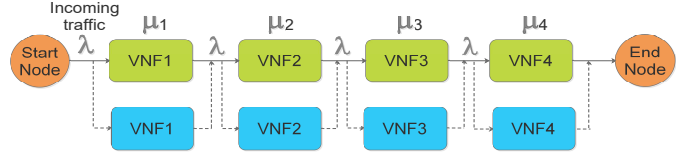}
	\caption{An SFC with single dedicated backup.} 				
	\label{Fig:scb}
\end{figure}

Although assigning redundant backup resources increases the reliability of service chains, this approach is inefficient with respect to efficient utilization of resources. The redundant backup resources are idle until a failure happens in the primary nodes or links. Also, since failure may happen randomly at any point of time, assigned redundant backup resources cannot be used for any other purposes. In this letter, we propose a novel solution called eRESERV which enhances the reliability of service chains
with less additional resources rather than assigning redundant
backup resources.

Here, instead of providing additional backups to VNFs in an SFC, in this letter, we propose to divide an SFC into multiple subchains of SFC with lower capacity VNFs to increase the reliability. 

\begin{theorem}
Dividing an SFC into subchains of SFC with VNFs of lesser capacity will increase the reliability of the system.
\label{thm:subChain}
\end{theorem}

\begin{proof}
Consider an SFC $s$ with arrival rate $\lambda_s$ is divided into $l$ subchains with each VNF $v \in \mathcal{V}_s$ having processing rate of $\frac{\mu_v}{l}$ and each subchain having an arrival rate of $\frac{\lambda_s}{l}$ as shown in Fig. \ref{Fig:ssc1}. However, as the reduced capacity/processing rate VNFs are performing the same software functionality as that of original VNF, the reliability of each VNF is still $p_v$. Let each subchain of $s$ be represented by $\{\bar{s}_1, \bar{s}_2, \ldots, \bar{s}_l\}$. Now, the reliability $r_{\bar{s}_i}$ of each subchain $\bar{s}_i$ can be calculated as,
\begin{equation}
r_{\bar{s}_i} = \prod\limits_{v \in \mathcal{V}_s} p_v \quad \forall i \in [1,l]
\end{equation}
However, for the system to be reliable, at least one of the subchains must be active. Therefore, reliability of the system $r_s$ can now be calculated as,
\begin{equation}
r_s= 1 - \prod\limits_{i=1}^{l}(1- r_{\bar{s}_i}) = 1 - \prod\limits_{i=1}^{l} (1 -  \prod\limits_{v \in \mathcal{V}_s} p_v) = 1 - (1 -  \prod\limits_{v \in \mathcal{V}_s} p_v)^l
\end{equation}
Differentiating the above equation with respect to $l$, we get
\begin{equation}
\frac{\text{d} r_s}{\text{d} l} = - (1- \prod\limits_{v \in \mathcal{V}_s} p_v)^l \times  \log (1 - \prod\limits_{v \in \mathcal{V}_s} p_v) > 0
\end{equation}
As $r_s$ is an increasing function with respect to $l$ and there do not exist any extreme points for all $l \in \mathbb{N}$, we can say that $r_s$ increases with increase in $l$. 
\end{proof}

\begin{figure}[t]
	\centering
	\includegraphics [width=\columnwidth, height=4.75cm]{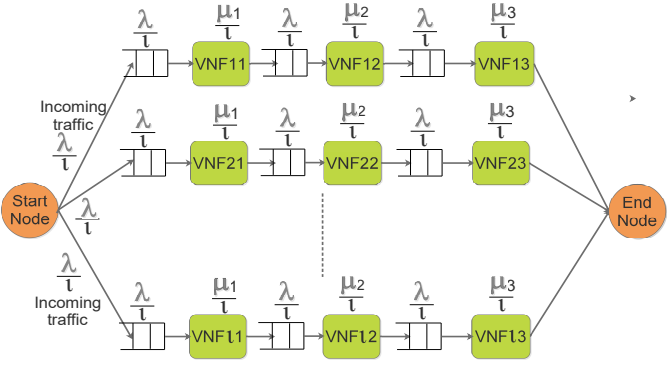}
	\caption{Subchaining as M/M/1 tandem network of queues.} 				
	\label{Fig:ssc1}
\end{figure}

\textit{Discussion: }
\begin{enumerate}
	\item 
	Although reliability of an SFC increases with increase in $l$, the reliability of this system is lesser than the backup setting. 
	\item
	However, the number of resources now required for an SFC will not increase, i.e., $\gamma_s = \sum\limits_{i=1}^{l}\sum\limits_{v \in \mathcal{V}_{\bar{s}_i}} \gamma_v = \sum\limits_{v \in \mathcal{V}_s} \gamma_v $ 
\end{enumerate}

\subsection{Analysing Delay}
From Theorem \ref{thm:subChain}, it is clear that, if an SFC is divided into multiple subchains, the reliability of the system increases. However, the new system with subchains must also meet the latency requirement $\Psi_s$ of the SFC $s$. To understand the effect of subchaining of an SFC on latency, we model subchaining of the SFC as a tandem of M/M/1 queuing network. By Burke's theorem \cite{Burke}, the arrival rates for all M/M/1 queues in the tandem of M/M/1 queuing network are same (refer Fig. \ref{Fig:ssc1}). 

Here, we model every VNF in a chain/subchain as an M/M/1 queue. The response time $\mathbb{E}[R]_v$ of the VNF $v$ can be calculated as,
\begin{equation}
\mathbb{E}_v[R] = \frac{1}{\mu_v - \lambda_v}
\end{equation}

\noindent where, $\lambda_v$ is the arrival rate to the VNF $v$. 

Consider an SFC $s$ with $|\mathcal{V}_s|$ number of VNFs. By Burke's theorem, $\lambda_v = \lambda_s~ \forall v \in \mathcal{V}_s$. The expected response time $\mathbb{E}_s[R]$ of the SFC can be calculated as,
\begin{equation}
\mathbb{E}_s[R] = \sum\limits_{v \in \mathcal{V}_s} \mathbb{E}_v[R] = \sum\limits_{v \in \mathcal{V}_s} \frac{1}{\mu_v - \lambda_s} \label{eq:delaySFC}
\end{equation}

Now, consider an SFC which is divided into $l$ subchains. As every packet traverses one of the subchains, say $\bar{s}$, the expected response time of the the SFC can be calculated as,
\begin{equation}
\mathbb{E}^{M/M/1}_s[R] = \sum\limits_{v \in \mathcal{V}_{\bar{s}}} \frac{1}{\frac{\mu_v}{l} - \frac{\lambda_s}{l}} = \sum\limits_{v \in \mathcal{V}_{\bar{s}}} \frac{l}{\mu_v - \lambda_s} \label{eq:delayMM1}
\end{equation}
Clearly, Equation $\ref{eq:delayMM1}$ is $l$ times of Equation $\ref{eq:delaySFC}$, showing that the response time of an SFC with subchains increases linearly when compared to original SFC without subchains. 

\subsection{Identifying Number of Subchains}
As every SFC $s$ has the delay constraint of $\Psi_s$, the delay incurred in an SFC with subchains should not exceed $\Psi_s$. Therefore,
\begin{align}
\sum\limits_{v \in \mathcal{V}_{\bar{s}}} \frac{l}{\mu_v - \lambda_s} \le \Psi_s 
\implies l \le \frac{\Psi_s }{\sum\limits_{v \in \mathcal{V}_{\bar{s}}} \frac{1}{\mu_v - \lambda_s}}
\label{equ:l_mm1}
\end{align}

\subsection{Decreasing Response Time}
In the previous subsection, we saw that the expected response time is linearly increasing with the number of subchains $l$. In this subsection we propose an alternate way to improve the response time. Here, we propose to have a common scheduler for every VNF as shown in Fig. \ref{Fig:sc-M/M/m} instead of having an individual scheduler for each VNF as in Fig. \ref{Fig:ssc1}. Now, the new system can be modeled as an M/M/m queuing system. Now, if a VNF is divided into $l$ smaller VNFs, then the expected response time of a VNF $v$ can be calculated as,
\begin{align}
\mathbb{E}_v^{M/M/l}[R] = \frac{l}{\mu_v} & \times \bigg(1+\frac{\varrho}{l(1-\frac{\lambda_v}{\mu_v})}\bigg) \\
%\end{align}    
%where,
%\begin{align*}
\text{where,}~~\varrho = \frac{(\frac{l\lambda_v}{\mu_v})^l}{l!(1-\frac{\lambda_v}{\mu_v})} & \times \bigg( \frac{1}{1+\frac{(\frac{l\lambda_v}{\mu_v})^l}{l!(1-\frac{\lambda_v}{\mu_v})}+\sum \limits _ {i=1}^{l-1} \frac{(\frac{l\lambda_v}{\mu_v})^i}{i!}} \bigg) \nonumber
\end{align}
Now, the expected response time of an SFC $s$ modeled as a tandem of M/M/m system can be calculated as,
\begin{align}
\mathbb{E}_s^{M/M/l}[R] = \sum\limits_{v \in \mathcal{V}_s} \frac{l}{\mu_v} \times \bigg(1+\frac{\varrho}{l(1-\frac{\lambda_s}{\mu_v})}\bigg)
\end{align}    
 
\subsection{Analyzing Reliability and Estimating l}
Here, the system will be active if any one of the smaller VNFs is active at every VNF of an original SFC.
Therefore, the reliability $r_s^{M/M/l}$ of the new M/M/$l$ system is given by,
\begin{align}
r_s^{M/M/l} =  \prod_{i=1}^{|\mathcal{V}_s|}(1 - (1 - p_{v})^l) 
\label{equ:sc_M/M/m_rel}
\end{align}

%\subsection{Estimating $l$}
As $\mathbb{E}_s^{M/M/l}[R] \le \Psi_s$ and $l \in \mathbb{N}$, $l$ can be calculated as,
\begin{equation}
l = \arg \min\limits_{\stackrel{i \in \mathbb{N}}{\Psi_s > \mathbb{E}_s^{M/M/i}[R]}} (\Psi_s - \mathbb{E}_s^{M/M/i}[R])
\label{equ:l_mml}
\end{equation}

\begin{figure}[t]
	\centering
	\includegraphics [width=\columnwidth, height=2cm]{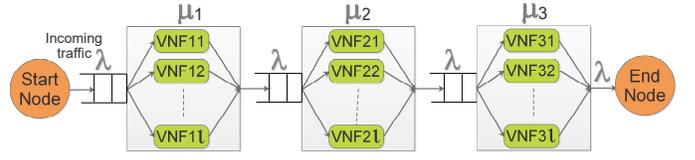}
	\caption{Subchaining as M/M/m tandem network of queues.} 				
	\label{Fig:sc-M/M/m}
\end{figure}

\begin{algorithm}[t]
	\small
	\caption{eRESERV solution}
	\label{algo:eReal}
	\hspace*{\algorithmicindent} \textbf{Input} ($\lambda_s, \mathcal{V}_s, \mu_v ~ \forall v \in \mathcal{V}_s, \Psi_s$) \\
	\hspace*{\algorithmicindent}  \textbf{Output} The number of subchains $l$ that can be created without violating $\Psi_s$
	\begin{algorithmic}[1]
		\If {M/M/1 setting}
		\State Calculate $l$ using Equation $\ref{equ:l_mm1}$
		\ElsIf{M/M/m setting}
		\State Calculate $l$ using Equation $\ref{equ:l_mml}$
		\EndIf
	\end{algorithmic}
\end{algorithm}

M/M/1 setting would be preferred in resource constrained systems and environments, and it reduces the migration delay as well when compared to M/M/m setting. Algorithm \ref{algo:eReal} gives the framework to identify the number of subchains in M/M/1 and M/M/m setting. If the preferred setting is M/M/1, then the number of subchains can be calculated in $\mathcal{O}(1)$ time. If the maximum number of subchains that can be made is $l_{max}$, then the number of subchains in M/M/m setting can be calculated in $\mathcal{O}(\log(l_{max}))$ time using binary search.

\section{Performance Analysis}
In this section, we evaluate the performance of our proposed solution methods presented in the previous section. Simulation parameters considered in our simulation are shown in Table \ref{tab:2}. Simulation results are obtained using discrete-event simulator MATLAB Simulink. We compare our proposed eRESERV M/M/1 and M/M/m settings with i) single SFC chain (SC) setting where there is one service chain for every service and ii) backup (SCB) setting where there is a dedicated backup for every VNF in an SFC. We compare our results in terms of reliability, expected response time, and number of resources required for an SFC. Note that, in the proposed M/M/1 and M/M/m settings, we first identify the number of subchains that can be created and then measure the performance metrics.

\begin{table}[h!]
  \begin{center}
  	\footnotesize
    \caption{Simulation parameters}
    \label{tab:2}
    \begin{tabular}{|l|l|}
      \hline 
      \textbf{Parameters} & \textbf{Values} \\
      \hline
      Arrival rate, $\lambda_s$ & 100 \\
      Serving rate of VNFs, $\mu_v$ & 200 \\
      Maximum allowed packet delay, $\Psi_s$ & 0.125 seconds\\
      Reliability rate of VNFs,  $p_{v}$ & 0.9 \\
      \hline
    \end{tabular}
  \end{center}
\end{table}

Fig. \ref{fig:subChansVRel} shows the effect of number of subchains created on the reliability. In SC setting, the reliability is always constant, however in SCB the reliability increases with increase in the number of backups. In any case, the reliability in eRESERV settings increases with increase in number of subchains (hence the title, ``the more the merrier"). M/M/m setting matches the reliability offered by the SCB setting, but consumes way less resources when compared to SCB setting as shown in Fig. \ref{fig:subChainsVRes}.

Fig. \ref{fig:VNFvNumOfSubchains} shows the number of subchains created by our proposed algorithm when the number of VNFs is varied in an SFC. As it can be seen, the number of subchains created decreases with an increase in the number of VNFs in an SFC in both the settings. However, the number of subchains created is more in M/M/m setting than in any other settings. This is due to the fact that, the expected response time in M/M/m setting is lesser than in M/M/1 setting when the number of subchains is increased as shown in Fig. \ref{fig:delayVsubchains}.

Fig. \ref{fig:delay} shows the effect of number of VNFs on the expected response time. As it can be seen, the proposed settings are able to meet the delay constraint at every instance of time. Although the response time in the proposed settings is higher than in SC and SCB settings, the reliability is the highest in M/M/m setting (with minimum amount of resources) when compared to all other settings as shown in Fig. \ref{fig:VNFVRel}.  

\begin{figure}[t]
	\begin{subfigure}{0.24\textwidth}
		\includegraphics[width=\textwidth]{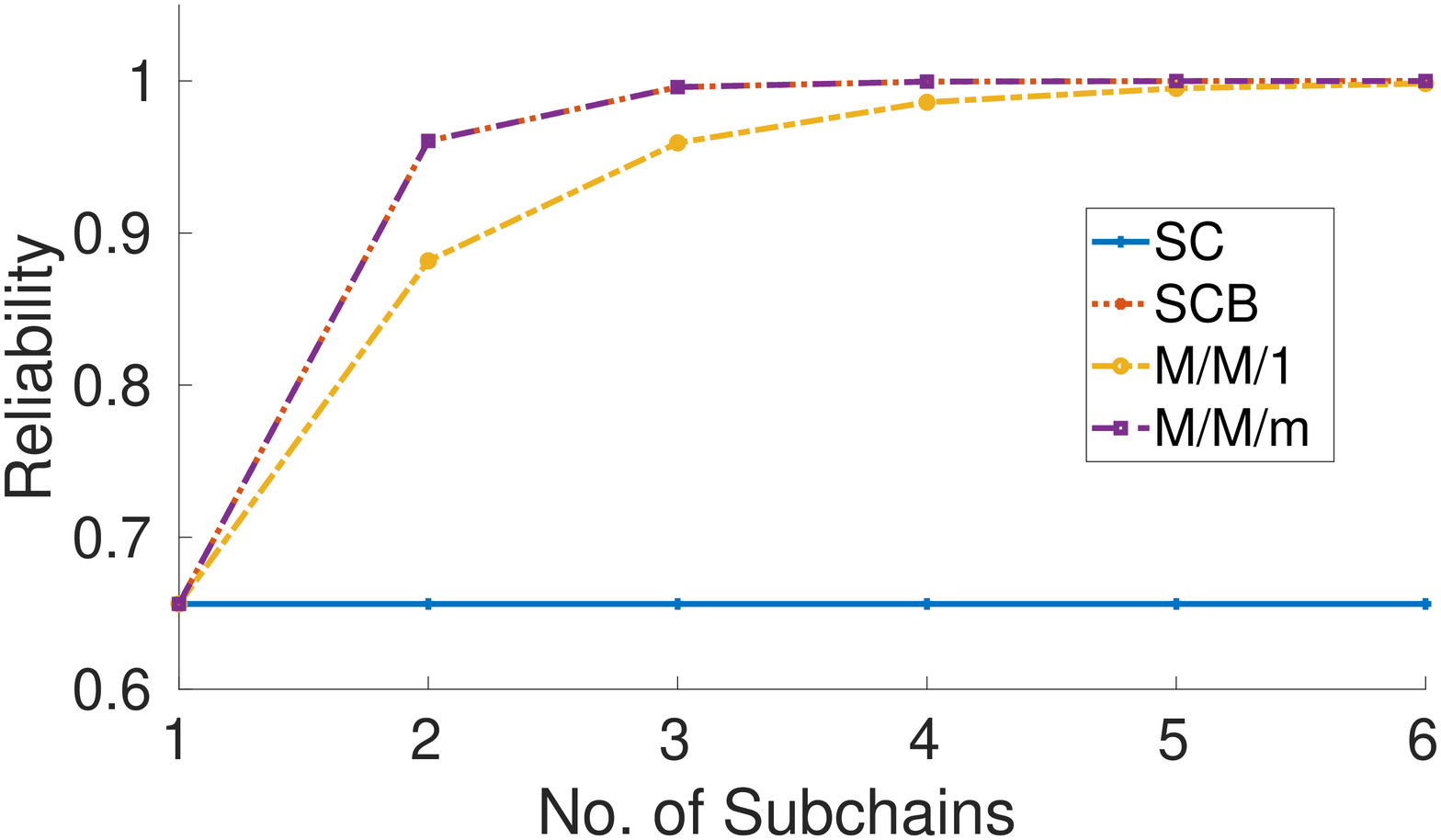}
		\caption{No. of subchains vs. Reliability}
		\label{fig:subChansVRel}
	\end{subfigure} 
	\begin{subfigure}{0.22\textwidth}
		\includegraphics[width=\textwidth]{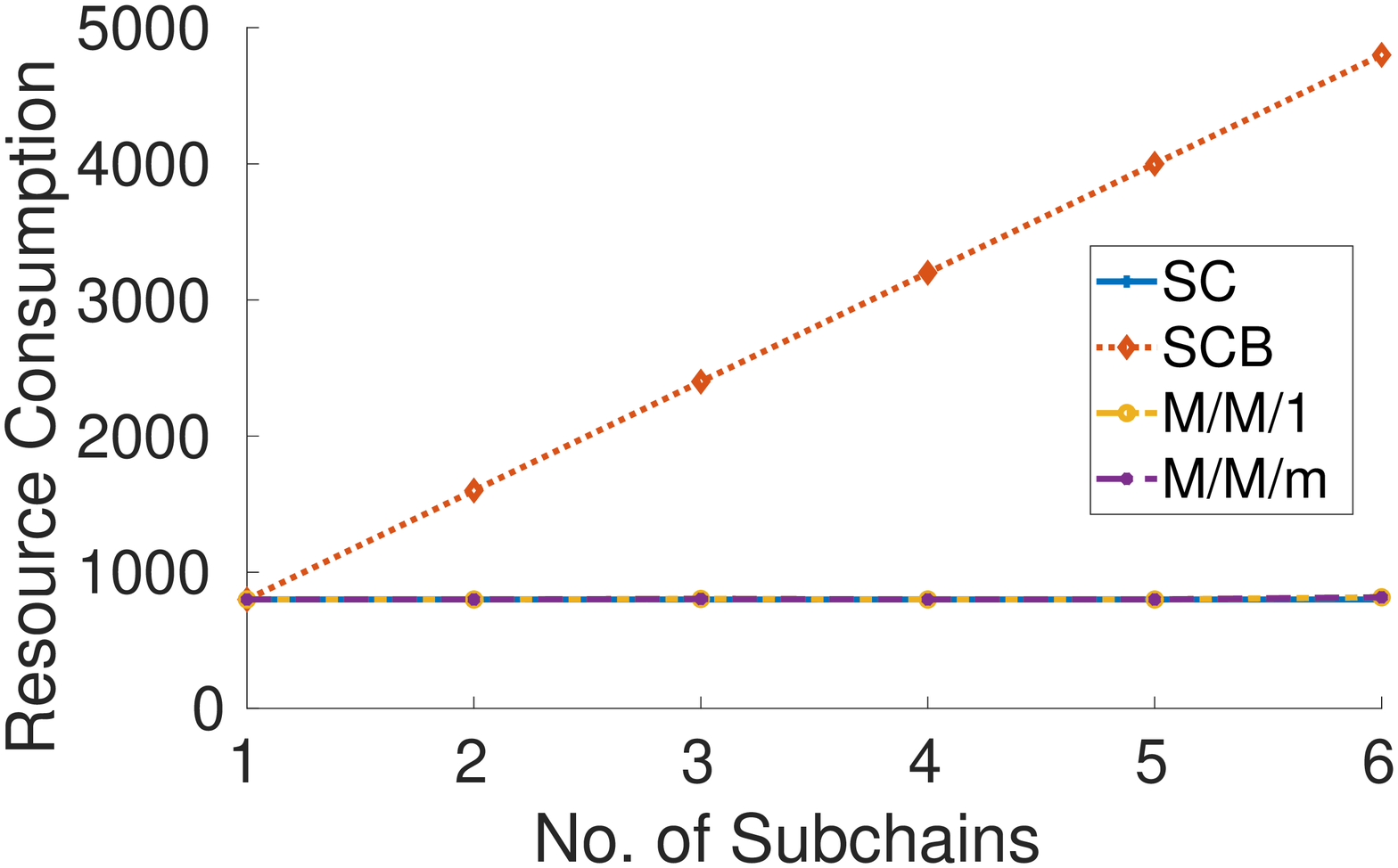}
		\caption{No. of subchains vs. Resource consumption}
		\label{fig:subChainsVRes}
	\end{subfigure}
	\begin{subfigure}{0.20\textwidth}
		\includegraphics[width=\textwidth]{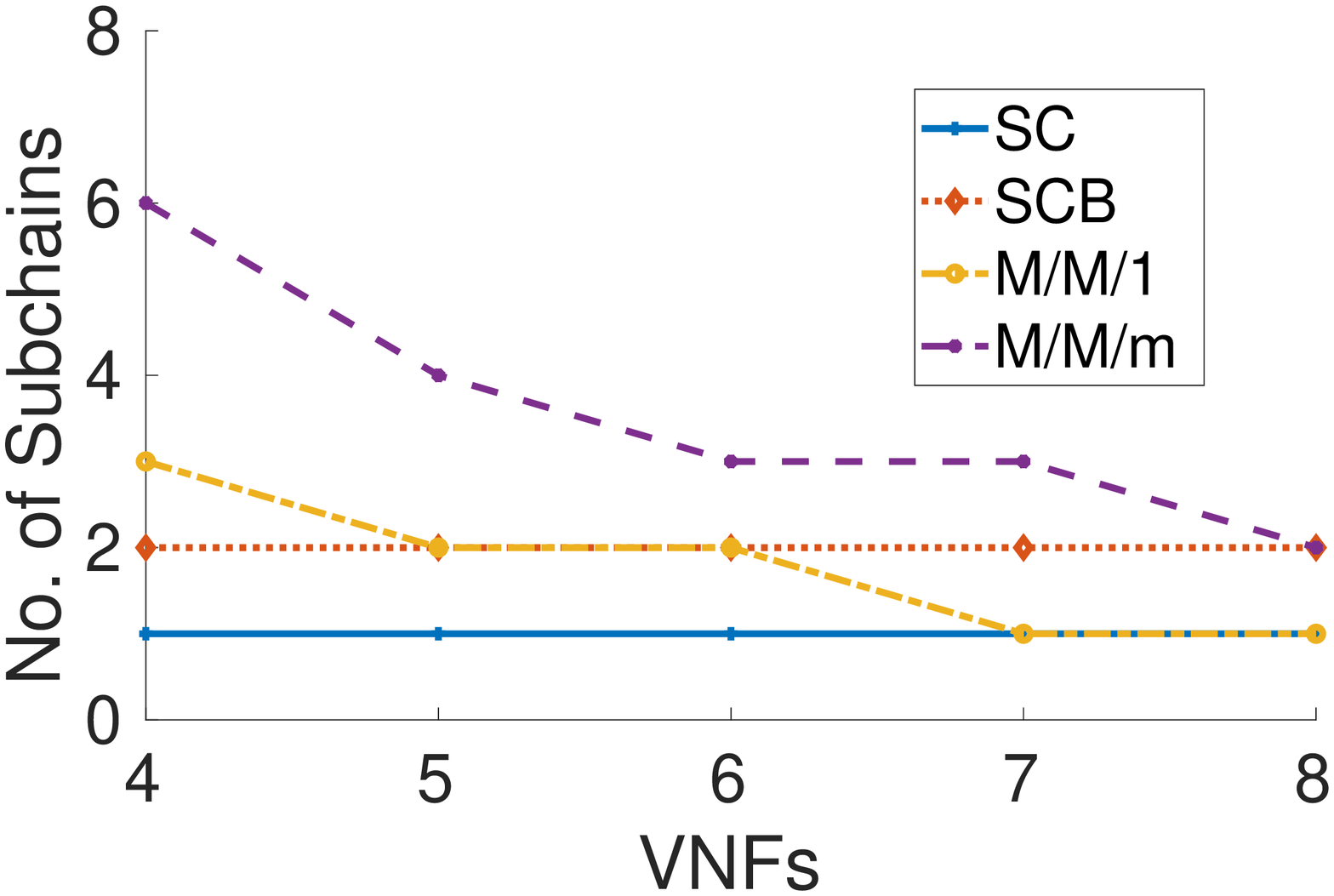}
		\caption{VNFs in an SFC vs. No. of subchains created}
		\label{fig:VNFvNumOfSubchains}
	\end{subfigure}
	\begin{subfigure}{0.22\textwidth}
		\includegraphics[width=\textwidth]{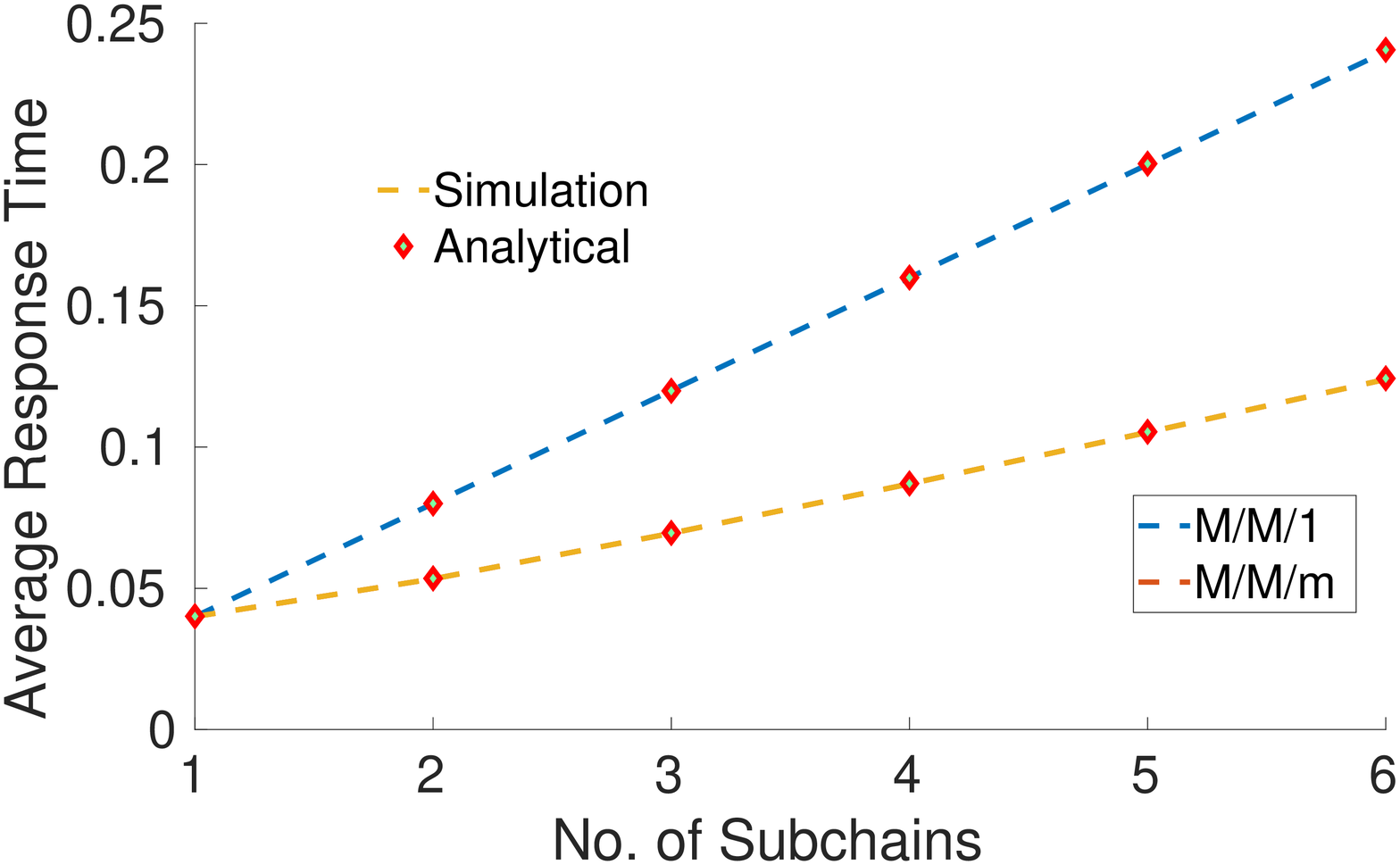}
		\caption{No. of subchains vs. Average response time}
		\label{fig:delayVsubchains}
	\end{subfigure}
	\begin{subfigure}{0.26\textwidth}
		\includegraphics[width=\textwidth]{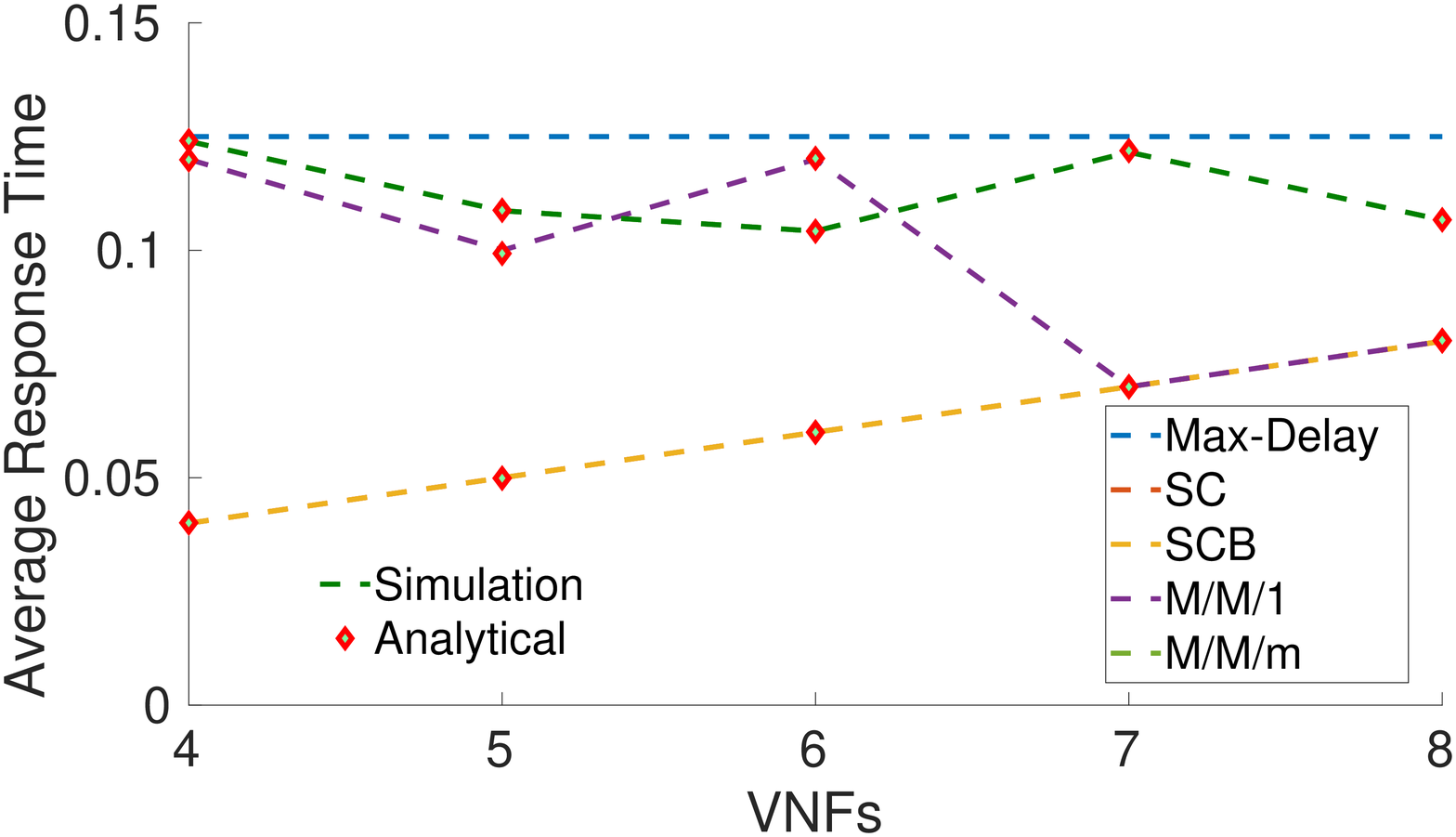}
		\caption{No. of subchains vs. Expected response time of an SFC}
		\label{fig:delay}
	\end{subfigure}
  \begin{subfigure}{0.21\textwidth}
    \includegraphics[width=\textwidth]{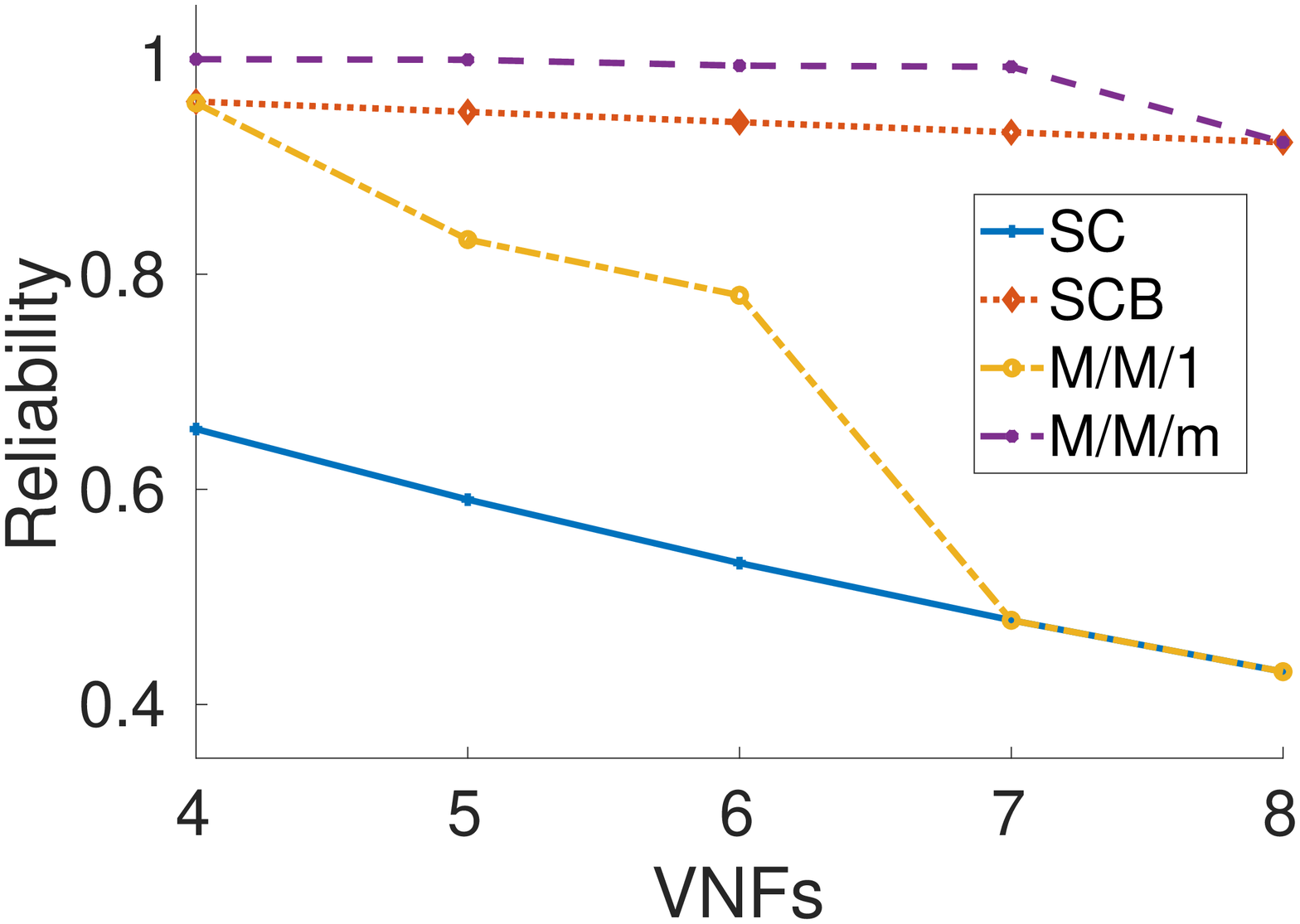}
    \caption{VNFs in an SFC vs. Reliability}
    \label{fig:VNFVRel}
  \end{subfigure}   
   \caption{Analytical and simulation results.}
\end{figure}

\section{Conclusion}
In this letter, we proposed a novel solution called eRESERV for enhancing the reliability of an SFC in 5G communication services. We first proposed utilization of subchains to enhance the reliability and decrease the amount of resource consumption. We analyzed this setting by modeling a VNF as an M/M/1 queue. Furthermore, to decrease the expected response time, we proposed a common scheduler for every VNF which was modeled as an M/M/m queue. Using queuing theory, we identified the number of subchains that can be created without violating the service delay constraints. By extensive simulations we showed the effectiveness of our proposed settings in terms of reliability, expected response time, and the amount of resources requested. 

In this letter, we considered that the substrate network is completely reliable. The placement of the proposed VNFs in an unreliable substrate network is an interesting study which we plan to pursue in our future work.

\ifCLASSOPTIONcaptionsoff
  \newpage
\fi

\bibliographystyle{IEEEtran}

\bibliography{references}

\end{document}